\documentclass[reqno,11pt]{amsart}
\usepackage{amssymb, amsmath}
\usepackage{amsthm, amsfonts,mathrsfs}
\usepackage{BOONDOX-cal}
\usepackage[T1]{fontenc}
\usepackage{times}
\usepackage[dvips]{graphics}
\usepackage{graphicx}
\usepackage{epsfig}
\usepackage{color}
\usepackage[all]{xy}
\newtheorem{theorem}{Theorem}[section]

\newtheorem{proposition}{Proposition}[section]
\newtheorem{remark}{Remark}[section]

\numberwithin{equation}{section}

\allowdisplaybreaks[4]

\begin{document}
\title[Equatorial stratified wind-drift currents with centripetal effects]
{Equatorial wind-drift currents with a continuous stratification and centripetal effects in the $\beta$-plane setting}%
\author[Lili Fan$^{\dag}$]{Lili Fan$^{\dag}$}%
\address[Lili Fan]{College of Mathematics and Information Science,
Henan Normal University, Xinxiang 453007, China}
\email{fanlily89@126.com\, (Corresponding author)}
\author[Ruonan Liu]{Ruonan Liu}%
\address[Ruonan Liu]{College of Mathematics and Information Science,
Henan Normal University, Xinxiang 453007, China}
\email{liuruonan97v@163.com}


\begin{abstract}
In this paper, we aim to derive an exact solution to the three-dimensional governing equations for wind-induced equatorial flows in the $\beta$-plane approximation with a depth-dependent density distribution and subjected to centripetal terms. The obtained explicit solution represents a steady purely azimuthal stratified flow with a flat surface and an impermeable flat bed that is suitable for describing the Equatorial Current. Resorting to a functional analysis, we show that the thermocline exhibits some monotonicity properties.
\end{abstract}

\date{}

\maketitle




\section{Introduction}
Consideration in this paper is the study of the nature of equatorial currents generated by the surface wind stress \cite{CJ3,WML}, noticeable for the near-surface current, and influenced by a continuous stratification and the centripetal force. Quite different from the Ekman theory \cite{Va}, investigated further recently in \cite{BC,Constantin2,DPC} etc., the equatorial currents exhibit no deflection of the surface current with respect to the wind due to the peculiarities of the Coriolis force at the Equator \cite{B,CI2,CJ4}. But based on the realistic consideration, such as its significant influence on the climate \cite{CJ,CJ3,Gill,TG}, the investigation of the nature of equatorial wind-stress currents is quite intricate and of great current interest. Besides, due to the importance of the depth of the thermocline in dynamic (c.f. El Nino phenomenon \cite{TPES}), we consider in the paper the azimuthal flows with a continuous stratification across the thermocline.

To describe the nonlinear dynamic of the given complex fluid flows in detail, it is remarkable to find an exact solution to the water wave problem. The Gerstner-type wave solution is well known and can be modified to describe a number of different physical and geophysical scenarios cf. \cite{CIKY2,Co1,Co2,Co3,Co4,CJ,CJ3,H1,I,Ma}. A most recent result for geophysical water waves with wind-stress can be referred to \cite{Co}. In regard to the research on the specifically azimuthal flows, the exact solutions have recently been studied by Constantin and Johnson in \cite{CJ2} for modelling of the homogeneous equatorial flows and the Antarctic Circumpolar Current (ACC) and subsequent studies on stratified flows can be referred to \cite{Ba,HM,HM1,HM2,HM3,HsM,M} ect. It is notable that the above studies are carried out in rotating spherical coordinates or in terms of cylindrical coordinates, which make the extraction of the fine details from the exact solutions difficult. To make a more apparent insight into the properties of the equatorial currents and the ACC, an alternative approach was pursued in \cite{CJ,MQ,Marynets,Q,YW2} to study homogeneous and stratified flows, where they simplified the geometry by relying on the equatorial $f$-plane approximation.

In the ocean dynamics, $f$-plane and $\beta$-plane approximation are two commonly used models \cite{CRB}. In the $f$-plane approximation, the Coriolis parameter is considered as constant, where the latitudinal variations are ignored and for the $\beta$-plane approximation, it introduces a linear variation with latitude of the Coriolis parameter. In this paper, we consider water waves in a moderate meridional distance from the equator, where the Navier-Stokes equations in the $\beta$-plane are applicative. Compared with the models investigated in \cite{Martin,Marynets,YW2}, the consideration of the $\beta$-plane approximation introduces a three-dimensional model with the appearance of the meridional coordinate. Besides, a higher precision with respect to an arbitrary stratification and the centripetal terms makes the three-dimensional model more complicated.

In this paper, we mainly obtain the exact solution of the linearised Navier-Stokes equations for a steady-state stratified flow with centripetal effects and under the assumption of a uniform wind stress. In our continuous stratification setting, the exact solution of the azimuthal velocity for the two-layer stratified fluid considered in \cite{Martin,YW2} can be recovered by taking a limiting process in our main integral formula for the azimuthal velocity field. The expression of the pressure obtained exhibits the three-dimensionality, which reduces to the two-dimensional results for the two-layer stratified fluid considered in \cite{Martin,YW2} and the continuous stratified fluid considered \cite{Marynets} without the consideration of the meridional coordinate. Besides, we confirm three monotonicity results by virtue of a functional analysis. Namely, the level of thermocline and strength of current velocity at thermocline decreases with the increase of the strength of wind speed at $10$ meters above the sea, and the strength of the flow reversal increases with the increase of the strength of wind speed at $10$ meters above the sea.

The remainder of this paper is organized as follows.  In Section 2, we present the governing equations for the equatorial wind-induced flow influenced by the centripetal force in the $\beta$-plane approximation. In Section 3, we derive the exact solution and give an analysis on the exact solution to obtain some monotonicity results.
\section{The governing equations in the $\beta$-plane setting with the centripetal terms}
Consideration in this paper is the effect of a uniform wind stress on the equatorial stratified water flows in a region of width of about 100 km, symmetric about the equator. In a reference frame with the origin located at a point fixed on Earth's surface and rotating with the Earth, we consider the zonal coordinate $x$ pointing east, the meridional coordinate $y$ pointing north and the vertical coordinate $z$ pointing up. Our aim is to derive purely azimuthal flow solutions, which means a steady flow moving in the azimuthal direction with vanishing meridional and vertical fluid velocity components. To this end, we consider the linearised Navier-Stokes in the $\beta$-plane approximation with the centripetal terms \cite{Co2,Henry3},
\begin{equation}\label{2.1}
\begin{cases}
&0=-\frac{1}{\rho}P_{x}+(\nu u_{z})_{z},\\
&\beta yu+\Omega^2 y=-\frac{1}{\rho}P_{y},\\
&-2\Omega u-\Omega^2 R=-\frac{1}{\rho}P_{z}-g,
\end{cases}
\end{equation}
and the mass conservation equation
\begin{equation}\label{2.2}
u_{x}=0,
\end{equation}
where $u$ is the horizontal fluid velocity component, $P=P(x,y,z,t)$ is the pressure field, $g=9.81$ m/s is the gravitational acceleration at the Earth's surface, $\rho=\rho(z)$ is a depth-depended density, $\nu=\nu(z)>0$ is the vertical eddy viscosity parameter \cite{CK}, $\Omega\approx7.29\times 10^{-5}$ rad s $^{-1}$ is the constant rotational speed of the Earth about the polar axis, $R=6371$ km is the radius of the Earth and $ \beta=2 \Omega / R=2.28 \times 10^{-11} \mathrm{~m}^{-1} \mathrm{~s}^{-1} $ (cf. \cite{CRB}).

Associated with the system \eqref{2.1}-\eqref{2.2} are the boundary conditions
\begin{equation}\label{2.3}
u=0 \quad \text{on} \quad z=-d,
\end{equation}
with $z=-d$ denoting the flat impermeable bottom, and
\begin{equation}\label{2.4}
u_{z}=0\quad \text{on} \quad z=-h,
\end{equation}
which denotes the shear vanishing at the thermocline $z=-h$ (cf. \cite{CK}). On the flat surface $z=0$, a discussion in \cite{CJ} and an assumption that the shear does not depend on the wind speed show that
\begin{equation}\label{2.5}
u_{z}=-\alpha \quad \text{on} \quad z=0,
\end{equation}
where $\alpha$ is a constant and is positive as the trade winds blow from east to west in the equatorial Pacific.
\section{Main results}

\subsection{Explicit solution for azimuthal velocity and pressure}
We adapt methods in \cite{Marynets} to bear relevance to our derivation for the system \eqref{2.1}-\eqref{2.5}.
\begin{theorem}\label{the3.1}
The solution of the purely azimuthal flow system \eqref{2.1}-\eqref{2.5} is given by
\begin{equation}\label{3.1}
u(z)=-\alpha \nu(0) \int^z_{-d}{\left(\frac {1} {\nu(s)}\frac{\int^s_{-h}{\frac{1}{\rho(r)}}dr}{\int^0_{-h}{\frac{1}{\rho(r)}}dr}    \right)}ds,
\end{equation}
and
\begin{align}\label{3.2}
P(x,y,z)=&-\frac{\alpha \nu(0)}{\int^0_{-h}{\frac{1}{\rho(s)}}ds}x
-2\Omega \int^0_{z}{\rho(s)u(s)}ds
-(\Omega^2 R-g)\int^0_{z}{\rho(s)}ds\nonumber\\
&-\rho(z)\left[\frac{\beta y^2}{2}u(z)+\frac{\Omega^2}{2}y^2\right]+P_{atm},
\end{align}
where $(x,y,z)\in \mathbb{R}^2\times [-d,0]$ and $P_{atm}$ is the constant atmospheric pressure.
\end{theorem}
\begin{proof}
Differentiating equation \eqref{2.1} with respect to $x$ and utilizing \eqref{2.2}, we obtain that
\begin{equation*}
(P_{xx},P_{xy},P_{xz})=\nabla P_{x}=0,
\end{equation*}
which shows that there is some constant $a$ such that
\begin{equation}\label{3.3}
P_{x}=a \quad \text{within the fluid domain}.
\end{equation}
Plugging \eqref{3.3} into the first equation of \eqref{2.1}, we get that
\[
(\nu u_{z})_{z}=\frac{a}{\rho(z)},
\]
which implies that there exists some function $b(y)$ such that
\begin{equation}\label{3.4}
\nu u_{z}=a\int^z_{-h}{\frac{1}{\rho(s)}}ds+b(y)\quad \text{for}\;-d\leq z\leq 0.
\end{equation}
On the other hand, an employment of the boundary conditions \eqref{2.4} and \eqref{2.5} yield that
\begin{equation*}
-\alpha \nu(0)=a\int^0_{-h}{\frac{1}{\rho(s)}}ds+b(y),\quad \text{and}\quad
\nu(-h)\cdot 0=0+b(y),
\end{equation*}
i.e.,
\begin{equation*}
a=-\frac{\alpha \nu(0)}{\int^0_{-h}{\frac{1}{\rho(s)}}ds},\quad b(y)=0.
\end{equation*}
Then we get from \eqref{3.4} that
\begin{equation}\label{3.5}
u_{z}=-\frac{\alpha \nu(0)}{\nu(z)} \frac{\int^z_{-h}{\frac{1}{\rho(s)}}ds}{\int^0_{-h}{\frac{1}{\rho(s)}}ds}.
\end{equation}
An integration of \eqref{3.5} with respect to $z$ leads to
\begin{equation}\label{3.6}
u(z)=-\alpha \nu(0) \int^z_{-d}\left({\frac{1}{\nu(s)} \frac{\int^s_{-h}{\frac{1}{\rho(r)}}dr}{\int^0_{-h}{\frac{1}{\rho(r)}}dr}}\right)ds,
\end{equation}
where the boundary condition \eqref{2.3} is used.

To get the expression of $P$, we integrate the first equation of \eqref{2.1} with respect to $x$ and the third equation of \eqref{2.1} with respect to $z$ to reach that
\begin{equation}\label{3.7}
P(x,y,z)=-\frac{\alpha \nu(0)}{\int^0_{-h}{\frac{1}{\rho(s)}}ds}x-2\Omega \int^0_{z}{\rho(s) u(s)}ds
-(\Omega^2R-g)\int^0_{z}{\rho(s)}ds+\tilde{p}(y),
\end{equation}
for some function $\tilde{p}(y)$. Substituting \eqref{3.7} into the second equation of \eqref{2.1}, we deduce that
\begin{equation}\label{3.8}
\frac{d\tilde{p}(y)}{dy}=-\rho(z)[\beta yu(z)+\Omega^2y].
\end{equation}
Letting $P_{atm}$ be the constant atmospheric pressure at the point $(x,y,z)=(0,0,0)$, we get that
\begin{equation}\label{3.9}
\tilde{p}(y)=-\rho(z)\left[\frac{\beta y^2}{2}u(z)+\frac{\Omega^2}{2}y^2\right]+P_{atm},
\end{equation}
which gives the expression of $P$ as in \eqref{3.2}.
\end{proof}

The expression of $u_z$ in \eqref{3.5} illustrates that the current speed increases strictly from flat surface to the level of thermocline and decreases strictly from the level of thermocline to the flat bottom, as shown in Fig. 1.
\begin{figure}[htbp]
\centering\includegraphics[width=3.4in]{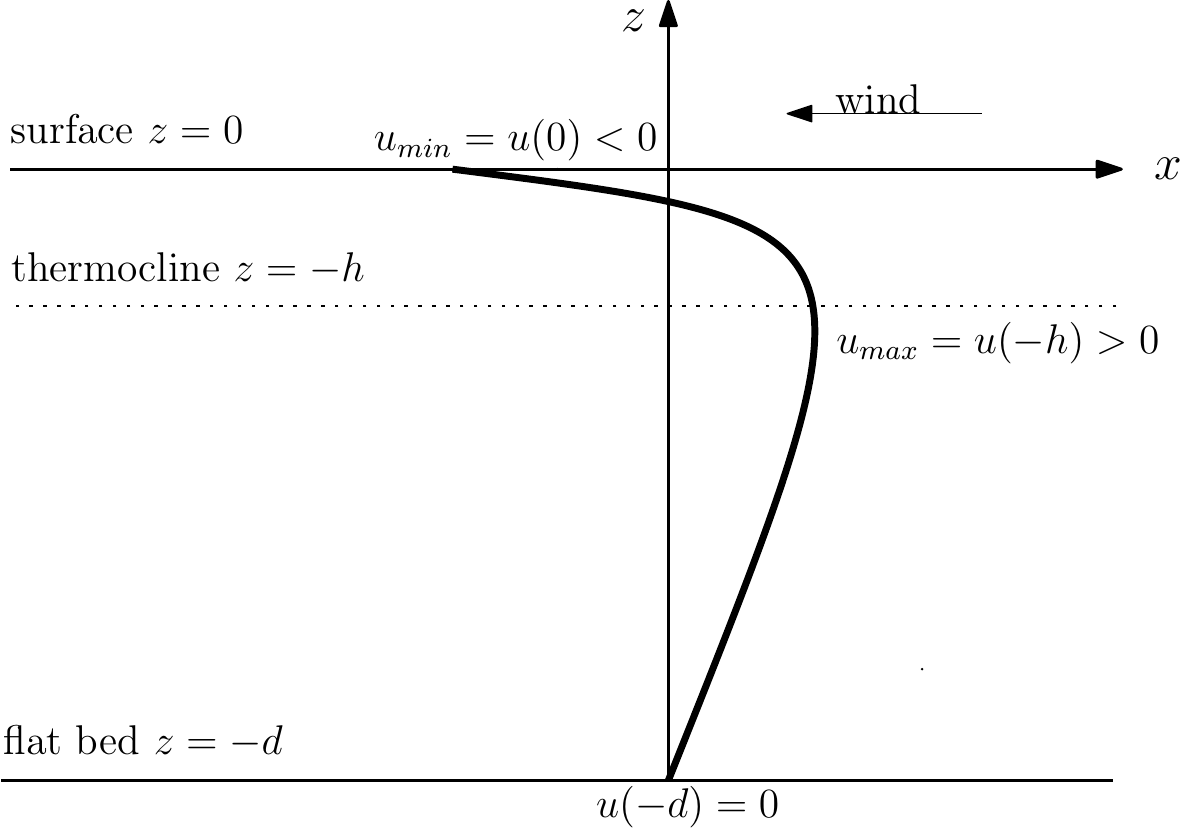}
\caption{Typical vertical profile of the current field in the equatorial Pacific region. The horizontal axis points from West to East, along the Equator.}\label{fig1}
\end{figure}
\begin{remark}The results in Theorem \ref{the3.1} are the three-dimensional generalizations of the ones in \cite{Marynets,YW2}. In fact, taking
\begin{equation*}
\rho(z)=
\begin{cases}
&\rho_0,\quad -h < z \leq 0, \\
&\rho_1,\quad -d \leq z < -h,
\end{cases}
\end{equation*}
where $\rho_0$ and $\rho_1$ are constant densities with $\rho_0 < \rho_1$, the equation \eqref{3.1} is reduced to
\begin{equation}\label{u}
u(z)=
\begin{cases}
&-\frac{\alpha \nu(0)\rho_0}{\rho_1} \int^z_{-d}{\frac{1+\frac{s}{h}}{\nu(s)}}ds\quad
\text{for}\; -d\leq z\leq -h,\\
&-\alpha \nu(0) \int^z_{-h}{\frac{1+\frac{s}{h}}{\nu(s)}}ds
-\frac{\alpha \nu(0)\rho_0}{\rho_1} \int^{-h}_{-d}{\frac{1+\frac{s}{h}}{\nu(s)}}ds \quad \text{for}\;-h\leq z\leq 0,
\end{cases}
\end{equation}
and the equation \eqref{3.2} is reduced to
\begin{equation}\label{P}
\begin{cases}P(x,y,z)&=-\frac{\alpha \nu(0)}{h}\rho_0 x
-2\Omega \rho_0 \int^0_{-h}{u(s)}ds
-2\Omega \rho_1 \int^{-h}_z{u(s)}ds
+\rho_0(g-\Omega^2 R)h\\
&-\rho_1(g-\Omega^2 R)(z+h)
-\rho_1 \left(\frac{\beta y^2}{2}u(z)+\frac{\Omega^2 y^2}{2}\right)+P_{atm}\qquad \text{for}\; -d\leq z\leq -h,\\
P(x,y,z)&=-\frac{\alpha \nu(0)}{h}\rho_0 x
-2\Omega \rho_0 \int^0_{z}{u(s)}ds
-\rho_0(g-\Omega^2 R)z\\
&-\rho_0 \left(\frac{\beta y^2}{2}u(z)+\frac{\Omega^2 y^2}{2}\right)+P_{atm} \qquad \text{for}\;-h\leq z\leq 0.
\end{cases}
\end{equation}
The expressions \eqref{u} and \eqref{P} coincide with the two-dimensional case considered in \cite{YW2} for $y=0$ and $\alpha=\frac 1 {\rho_0 \sigma}$.

Taking no account of the effect of the centripetal force, the expression of $P$ is given by
\begin{align*}
P(x,y,z)=-\frac{\alpha \nu(0)}{\int^0_{-h}{\frac{1}{\rho(s)}}ds}x
-2\Omega \int^0_{z}{\rho(s) u(s)}ds
+g\int^0_{z}{\rho(s)}ds
-\frac{\beta y^2}{2}\rho(z) u(z)+P_{atm}.
\end{align*}
This coincides with the two-dimensional case considered in \cite{Marynets} for $y=0$.
\end{remark}
\vspace{-4mm}
\subsection{Monotonicity}
In this subsection, the method in \cite{Martin} is adapted our analysis on the monotonicity of the level of the subsurface and the strength of the current at the subsurface in connection with the strength of the wind speed near the ocean's surface.

Employing the arguments in \cite{CJ,CK,PGT,SHM}, we make the assumption that the viscosity has the form
\begin{equation}\label{3.10}
\nu(z)=\nu(0) f\left(\frac{z}{d}\right), \quad -d \leq z \leq 0,
\end{equation}
where a suitably chosen $f:\ [-1,0]\rightarrow (0,\infty)$ is a decreasing function with depth in the layer above the thermocline. On the other hand, the discussions in \cite{BVW,GOOD,Martin} lead us to grasp the relation between the wind speed $U_{wind}$ at $z$ meters above the sea and the velocity of the free surface $u(0): =u_{0}$ as follows
\begin{equation}\label{3.11}
U_{wind}=\frac{1}{\kappa} u_{0}\ln\left(\frac{zg}{au_{0}^2} +1\right),
\end{equation}
where $\kappa$ is known as the $K\acute{a}rman$ constant and $a$ is a positive constant. It is usually typical to consider the wind speed at $10$ meters above the sea, denoted as $U_{10}$. On account of the wind blowing from the east to the west, we have $U_{wind} < 0$, resulting in $u_{0} < 0$.

Now, we are in the position to introduce the monotonicity between the strength of wind speed at $10$ meters above the sea, $|U_{10}|$, and the level of the thermocline $-h$.
\begin{proposition}\label{pro3.1}
We assume the eddy viscosity function is given by \eqref{3.10}. Then the level of the thermocline $-h$ decreases as the strength of wind speed $|U_{10}|$ increases.
\end{proposition}
\begin{proof}Combining \eqref{3.1} with \eqref{3.10}, we infer that
\begin{equation*}
u_{0}=-\alpha \int^0_{-d}{\left(\frac{1}{f\left(\frac{s}{d}\right)} \frac{\int^s_{-h}{\frac{1}{\rho(r)}}dr}{\int^0_{-h}{\frac{1}{\rho(r)}}dr} \right)}ds < 0,
\end{equation*}
where the sign of $u_{0}$ is due to the reality that the trade wind blows westwords in the equatorial Pacific. By a calculation
\begin{align}\label{3.12}
\frac{d|u_{0}|}{dh}&=\alpha \int^0_{-d}{\frac{1}{f\left(\frac{s}{d}\right)}
\frac{\partial}{\partial h}\left(\frac{\int^s_{-h}{\frac{1}{\rho(r)}}dr}
{\int^0_{-h}{\frac{1}{\rho(r)}}dr}\right)}ds\nonumber\\
&=\alpha \int^0_{-d}{\frac{1}{f\left(\frac{s}{d}\right)}
\frac{\frac{1}{\rho(-h)}\left(\int^0_{-h}{\frac{1}{\rho(r)}}dr-\int^s_{-h}{\frac{1}{\rho(r)}}dr\right)}
{\left(\int^0_{-h}{\frac{1}{\rho(r)}}dr\right)^2}}ds\nonumber\\
&=\alpha \int^0_{-d}{\frac{1}{f\left(\frac{s}{d}\right)}
\frac{\frac{1}{\rho(-h)}\int^0_{s}{\frac{1}{\rho(r)}}dr}
{\left(\int^0_{-h}{\frac{1}{\rho(r)}}dr\right)^2}}ds > 0,
\end{align}
we conclude that $|u_{0}|$ exists invertible function about $h$, thus
\begin{equation}\label{3.13}
\frac{dh}{d|u_{0}|}=\left(\frac{d|u_{0}|}{dh}\right)^{-1} > 0.
\end{equation}
By \eqref{3.11},
\begin{equation}\label{3.14}
|U_{10}|=\frac{1}{\kappa} |u_{0}|\ln\left(\frac{10g}{au_{0}^2} +1\right).
\end{equation}
Then we get that
\begin{equation}\label{3.15}
\frac{d|U_{10}|}{d|u_{0}|}=\frac{1}{\kappa}\left[\ln\left(\frac{10g}{au_{0}^2}+1\right)
-\frac{20g}{10g+au_{0}^2}\right].
\end{equation}

Our next step is to determine the sign of \eqref{3.15}. Noting that $0< \frac{20g}{10g+au_{0}^2} < 2$, i.e., $e^{\frac{20g}{10g+au_{0}^2}}<e^2$ for all $|u_{0}|>0$, we obtain $\frac{10g}{au_{0}^2}+1>e^2$ for $|u_{0}|< \left[\frac{10g}{a(e^2-1)}\right]^{\frac{1}{2}}$, which implies the right side of \eqref{3.15} is strictly positive for $|u_{0}|\in(0,\left[\frac{10g}{a(e^2-1)}\right]^{\frac{1}{2}})$. Then, we infer that $|U_{10}|$ exists invertible function about $|u_{0}|$ for all $|u_{0}|\in \left(0, \left[\frac{10g}{a(e^2-1)}\right]^{\frac{1}{2}}\right)$, i.e.,
\begin{equation}\label{3.16}
\frac{d|u_{0}|}{d|U_{10}|}=\left(\frac{d|U_{10}|}{d|u_{0}}\right)^{-1} > 0.
\end{equation}
Due to \eqref{3.13} and \eqref{3.16}, it is obvious that
\begin{equation}\label{3.17}
\frac{dh}{d|U_{10}|}=\frac{dh}{d|u_{0}|} \frac{d|u_{0}|}{d|U_{10}|} >0.
\end{equation}
Therefore, the conclusion is confirmed.
\end{proof}

The monotonicity between the current at the thermocline $u(-h)$ and the strength of the wind speed $|U_{10}|$ is presented as follows.
\begin{proposition}\label{pro3.2}
The strength of the current at the thermocline $u(-h)$ decays as the strength of the wind speed $|U_{10}|$ increases. Besides, the difference $u(-h)-u(0)$, measuring the strength of the flow reversal, increases as $|U_{10}|$ increases.
\end{proposition}
\begin{proof}
By \eqref{3.1}, we have
\begin{equation}\label{3.18}
\frac{d(u(-h))}{dh}=-\frac{\alpha}{\rho(-h)\left(\int^0_{-h}{\frac{1}{\rho(r)}}dr\right)^2}
\int^{-h}_{-d}{\frac{\int^0_{s}{\frac{1}{\rho(r)}}dr}{f\left(\frac{s}{d}\right)}}ds < 0,
\end{equation}
and due to the relation
\begin{equation*}
u(-h)-u(0)=\alpha \int^0_{-h}{\frac{1}{f\left(\frac{s}{d}\right)} \frac{\int^s_{-h}{\frac{1}{\rho(r)}}dr}{\int^0_{-h}{\frac{1}{\rho(r)}}dr}}ds,
\end{equation*}
we have
\begin{equation}\label{3.20}
\frac{d[u(-h)-u(0)]}{dh}=\frac{\alpha}{\rho(-h)\left(\int^0_{-h}{\frac{1}{\rho(r)}dr}\right)^2}
\int^0_{-h}{\frac{\int^0_{s}{\frac{1}{\rho(r)}}dr}{f\left(\frac{s}{d}\right)}}ds > 0.
\end{equation}
Then we obtain from \eqref{3.17} that
\begin{equation}\label{3.19}
\frac{du(-h)}{d|U_{10}|}=\frac{du(-h)}{dh} \frac{dh}{d|U_{10}|} < 0,
\end{equation}
and
\begin{equation}\label{3.21}
\frac{d[u(-h)-u_{0}]}{d|U_{10}|}=\frac{d[u(-h)-u_{0}]}{dh} \frac{dh}{d|u_{10}|} >0.
\end{equation}
This completes the proof of Proposition \ref{pro3.2}.
\end{proof}
\vspace{0.5cm}
\noindent {\bf Acknowledgements.}
The work of Fan is partially supported by a NSF of Henan Province of China Grant No. 222300420478 and the NSF of Henan Normal University Grant No. 2021PL04.

\end{document}